\RequirePackage[T1]{fontenc} \documentclass[journal]{IEEEtran} \usepackage{amsmath,amssymb,amsthm,bm} \usepackage{graphicx,xcolor,booktabs} \usepackage{cite} \usepackage[hidelinks]{hyperref} \newtheorem{theorem}{Theorem}  \newtheorem{corollary}[theorem]{Corollary} \newtheorem{lemma}[theorem]{Lemma} \theoremstyle{remark} \newtheorem*{example}{Example} \newcommand{\R}{\mathbb R} \newcommand{\C}{\mathbb C} \newcommand{\tr}{\operatorname{tr}} \newcommand{\rank}{\operatorname{rank}} \newcommand{\diag}{\operatorname{diag}}    \newcommand{\smin}{\sigma_{\min}}  \newcommand{\eps}{\varepsilon} \newcommand{\cD}{\mathcal D} \newcommand{\marginrec}{r_{\mathrm{rec}}} \newcommand{\appref}[1]{Appendix~\ref{#1}} \definecolor{paperblue}{HTML}{0080FF} \definecolor{paperred}{HTML}{FF6666} \hbadness=10000\vbadness=10000\relax\begin{document} \title{Minimal Experiments for Robust Stabilization: Information, Spectral Geometry, and Duration} \author{Alexey~Peregudin and Ngoc~Tuan~Dinh\thanks{Corresponding author: Alexey Peregudin.}\thanks{Alexey Peregudin is with the School of Electrical and Electronic Engineering, University of Sheffield, Sheffield, United Kingdom (e-mail: a.peregudin@sheffield.ac.uk).}\thanks{Ngoc Tuan Dinh is with ITMO University, St. Petersburg, Russia (e-mail: dinhngoctuan6789@gmail.com).}} \maketitle\par
\par
\begin{abstract} On broad classes of linear systems, the shortest experiments are almost as good as the best possible ones. For $n$ states and $m$ inputs, the shortest input sequences that support robust data-driven stabilization of every controllable plant have $mn+1$ steps with exact states and $m(n+1)$ with noisy states. We show that, when the spectral radius is bounded and the spectrum is well separated near the unit circle, these sequences tolerate a fixed fraction of the error level achievable by any experiment, even one designed with full plant knowledge and allowed to use any finite duration. This constant-factor comparison can fail for slowly actuated systems. For $A=I+hG$ with controllability depth $\nu\ge2$, short experiments lose a factor of order $h^{\nu-1}$, and duration of order $1/h$ is both necessary and sufficient to recover a fixed fraction of the optimal tolerance. \end{abstract}  \section{Introduction}\label{sec:intro} Choosing an experiment for control requires more than collecting enough data to compute a feedback gain. A short record may contain all the information needed for stabilization while remaining sensitive to small errors. Additional samples can reveal missing state directions or allow weak responses to become more pronounced. Whether they improve robustness is therefore a quantitative question: knowing that an experiment works for sufficiently small errors does not tell us how much error it can tolerate. \par
We study experiments that start from a known zero state and use bounded inputs. Process and measurement errors are bounded uncertainties, without a probability model. Their bounds are normalized by the square root of duration, giving short and long experiments a common scale for comparing noise tolerance. \par
We measure an experiment by the error level it tolerates on \emph{every} possible data record: each record must provide one static state-feedback gain and one quadratic Lyapunov function valid for all models consistent with it. We compare experiments with a benchmark, the \emph{causal oracle}. The oracle knows the true plant and observes exact states; it fixes the duration before collecting data and chooses each input from the observations collected so far. Its records must meet the same requirement. A short experiment is \emph{competitive} if its error tolerance is at least a fixed positive fraction of the oracle's, with the same fraction valid throughout the stated plant class. The comparison thus measures what plant knowledge, feedback during data collection, and additional time are worth for this control task. \par
\subsection{Principal results} \par
The main conclusions are as follows. \par
\begin{itemize} \item\emph{Minimum information.} The exact universal minimum duration for predetermined inputs is $mn+1$ transitions with exact states and $m(n+1)$ with any positive measurement uncertainty. The increase comes from a stronger rank requirement induced by sensor error. \par
\item\emph{Minimum duration is often enough.} Under bounded spectral radius and suitable spectral separation near the unit circle, the shortest universal experiments achieve a constant fraction of the robustness of the best plant-informed causal experiment, even though the latter observes exact states and may use any finite duration. No matrix-norm or eigenvector-conditioning bound is required. \par
\item\emph{But some plants need time, not merely information.} For $A=I+hG$ with controllability depth $\nu\ge2$, the shortest-experiment tolerance scales as $h^{\nu-1}$ while the benchmark remains order one. Constant-factor competitiveness requires and can be achieved with duration of order $1/h$. \end{itemize} \par
The last result separates weak controllability from slow controllability: in some families the same small controllability scale limits every experiment, whereas in others it only reflects that the shortest experiment ends before the relevant dynamics have had time to develop. \par
\subsection{Related work and scope} For controllable systems, the fundamental lemma represents all trajectories of a prescribed finite length using one trajectory with sufficiently rich inputs~\cite{willems2005}. This supports data-based feedback design~\cite{depersis2020}. The informativity framework makes the set of models consistent with the data explicit and shows that stabilization need not require identification~\cite{informativity2020}. For a known zero initial state, our exact-state result reduces the stabilization criterion of~\cite{informativity2020} to the rank condition $\rank X=n$ and adds robustness to small process errors. For error-free data, \cite{fragility2026} links a failure of identification to vanishing robustness against perturbations of the implemented gain. In our setting, sensor errors instead change the possible underlying trajectories, and hence the set of models that the gain must stabilize. Even so, our rank condition does not require identifying input directions that the experiment never uses. \par
Feedback gains and common quadratic Lyapunov functions can be computed from noisy data using the matrix S-lemma~\cite{slemma2022}, matrix ellipsoids and Petersen's lemma~\cite{petersen2022}, and quadratic matrix inequalities~\cite{qmi2023}. Measurement errors have been treated by linear matrix inequalities~\cite{bisoffi2024} and polynomial optimization~\cite{eiv2025}. General data-perturbation formulations also provide certificates for combined and temporally structured errors~\cite{perturbation2025}. We use these tools to connect data errors with what an experiment can guarantee and to compute the required certificates. A common stabilizing gain need not admit a common quadratic Lyapunov function~\cite{nonquadratic2022}; requiring such a function is a restriction of our criterion. The proofs also use standard robust-control results~\cite{polik2007,boyd1994}. \par
Online experiment design~\cite{online2022} and the shortest input--output identification experiments of~\cite{shortest2025} choose inputs using previous observations. Adaptive partial identification can also suffice for error-free stabilization under a stabilizability prior~\cite{fast2023}. Universal-input results relate persistent excitation to plant-independent guarantees~\cite{universal2025}; prior knowledge can alter identification and stabilization requirements~\cite{prior2025}. Sparse pulse inputs have been used for persistent excitation~\cite{impulses2023} and stabilizing exploration~\cite{blackbox2021}. Our measured-state sequence uses the pulse spacing of~\cite{impulses2023}, with its initial zero segment omitted because the state starts at zero. Our exact minimum lengths concern inputs fixed before the experiment, a known zero initial state, and positive error tolerance on every controllable plant. For multiple inputs, shorter experiments that adapt to observations are not ruled out. The causal oracle is the robustness benchmark, not the subject of these minimum-length claims. \par
Quantitative notions of persistency of excitation bound the singular values of the data~\cite{quantitative2023}. Bounded-error experiment design studies the geometry of the set of parameters consistent with the data~\cite{bounded1989} and worst-case identification error under input constraints~\cite{worstcase2014}; recent universal designs address finite-horizon identification with general prior knowledge~\cite{setmembership2026}. Targeted exploration with energy-bounded disturbances seeks prescribed parameter accuracy~\cite{nonstochastic2026}. Stochastic task-dependent exploration and learning-to-control hardness results provide complementary perspectives~\cite{targeted2021,hard2022}. We compare experiments by the error level at which every record remains certifiable, with bounded inputs and operator-norm error bounds on the whole record, normalized by the square root of the duration. Short and long experiments are compared on the same true plant against a plant-informed causal benchmark. Spectral and controllability assumptions constrain the true plant, not the compatible alternatives. \par
The paper proceeds from the information needed for positive robustness, to the shortest universal experiments, to their competitiveness, and finally to the situations that force a longer duration. Section~\ref{sec:model} defines the criterion; Section~\ref{sec:experiments} establishes the information and minimum-duration results; Section~\ref{sec:spectral} gives the record certificates and the comparison with the oracle; Section~\ref{sec:duration} proves the duration law; and Section~\ref{sec:examples} contrasts slow actuation with a weak link and with nonnormality. Section~\ref{sec:conclusions} concludes. Brief proof ideas accompany the main results; complete proofs are in the Appendix. \par
\subsection{Notation} We use $\R$ and $\C$ for the real and complex fields, $I$ for an identity matrix of the required size, and $e_j$ for a coordinate vector. The symbols $M^\top$, $M^*$, and $M^\dagger$ denote the transpose, conjugate transpose, and Moore--Penrose inverse. Matrix norms are operator norms unless marked by $F$ (Frobenius); vector norms are Euclidean unless marked by $\infty$. We use the least row singular value $\smin(M):=\sqrt{\lambda_{\min}(MM^*)}$, which is zero when the rows are dependent, and write $\rho(M)$ for spectral radius. The notation $[M;N]$ means vertical stacking, $\otimes$ denotes the Kronecker product, and $\operatorname{im}M$ and $\operatorname{row}M$ denote column and row spaces. For real matrices, $P\succ0$ ($P\succeq0$) means symmetric positive (semi)definiteness; for complex matrices, the Hermitian convention applies. In a stated limit, for nonnegative $f$ and positive $g$, $f\asymp g$ means that $f/g$ is bounded above and below by positive constants; $O(g)$ and $\Omega(g)$ denote upper and lower bounds up to a positive constant. For vector or matrix quantities, $O(g)$ bounds the norm. Constants may depend on the fixed parameters specified in each statement, but not on the parameter tending to its limit. \par
For trajectory data, $X,Y$ are the true current and successor state blocks; they also denote arbitrary data matrices in the algebraic record tests. Tildes mark measured quantities; on state and data blocks, a superscript $0$ marks error-free (ideal) quantities. Hats mark least-squares fits, and primes mark a compatible alternative and its states and errors. Bars on states, system matrices, and data denote coordinates normalized for analysis; a subscript $c$ denotes coordinates computed from observations. On complex scalars, a bar denotes complex conjugation. \section{Problem and robustness criterion}\label{sec:model} Consider the real $n$-state, $m$-input plant \begin{equation}\label{eq:system} x_{t+1}=Ax_t+Bu_t+w_t,\quad\widetilde x_t=x_t+v_t,\quad x_0=0. \end{equation} The initial state is known exactly. An experiment has $T$ transitions and satisfies $\|u_t\|_\infty\le1$. Write $X=[x_0\ \cdots\ x_{T-1}]$, $Y=[x_1\ \cdots\ x_T]$, $U=[u_0\ \cdots\ u_{T-1}]$, $W=[w_0\ \cdots\ w_{T-1}]$, and $V=[v_0\ \cdots\ v_T]$. The measured blocks are $\widetilde X=X+V_-$ and $\widetilde Y=Y+V_+$, where $V_-$ and $V_+$ are the first and the last $T$ columns of $V$. Their shared columns come from the same sensor-error array. The separate error bounds \begin{equation}\label{eq:budgets} \|W\|\le\delta_w,\quad\|V\|\le\delta_v \end{equation} apply to the full error arrays. The state is not reset during the experiment, and no probability model or bound on the state magnitude is assumed. \par
Write $\cD=(\widetilde X,\widetilde Y,U)$ for the observed record. Its \emph{compatibility set} $\Sigma_{\delta_w,\delta_v}(\cD)$ consists of all real pairs $(A',B')$ that explain the record through some states $x'_0=0,x'_1,\ldots,x'_T$ and error arrays $W',V'$ satisfying the bounds \eqref{eq:budgets}: the dynamics $x'_{t+1}=A'x'_t+B'u_t+w'_t$ hold for $0\le t<T$, and $\widetilde x_t=x'_t+v'_t$ for $0\le t\le T$. With exact states write $\Sigma_{\delta_w}(X,Y,U):=\Sigma_{\delta_w,0}(X,Y,U)$; explicitly, \begin{equation}\label{eq:consistency} \Sigma_{\delta_w}(X,Y,U) =\{(A',B'):\|Y-A'X-B'U\|\le\delta_w\}. \end{equation} We also use the right-hand side of \eqref{eq:consistency} for arbitrary matrices $X,Y,U$ that need not come from one trajectory; fitted records require this algebraic extension. Alternative models range over all pairs of the prescribed dimensions: assumptions on the true plant are used to assess experiments, not to restrict this set. \par
A record is \emph{certifiable} if its compatibility set is nonempty and there are a real gain $K$ and $P\succ0$ such that \begin{equation}\label{eq:cq} (A'+B'K)^\top P(A'+B'K)-P\prec0 \end{equation} for every compatible pair; we then call $(K,P)$ a \emph{certificate} for the record. The certificate may depend on the observed record and the error bounds, but the same $K$ and $P$ must work for every compatible model. The input bound and the sensor errors concern data collection only: the feedback $u=Kx$ applied after the experiment uses the state and is not subject to saturation. \par
\subsection{Experiment tolerance and the causal oracle} A \emph{predetermined input sequence} is fixed before data collection. A deterministic causal policy $\pi$ may instead choose $u_t$ from observations available through time $t$. The duration $T$ is chosen before observations are collected. Fix a finite error ratio $\kappa\ge0$ and compare experiments using the error bounds \begin{equation}\label{eq:budget-ray} \delta_w=\eps\sqrt T,\quad\delta_v=\kappa\eps\sqrt T. \end{equation} Thus $\kappa=0$ means exact states, while $\eps$ is an error level normalized by the square root of duration. At a fixed $\eps$, the error bounds on the whole record grow as $\sqrt T$; short and long experiments use this same normalization. For example, if $\|w_t\|\le\eps$ at every step, then $\|W\|\le\|W\|_F\le\eps\sqrt T$. A fixed per-step process-error level therefore remains admissible at every duration. \par
A \emph{physical record} is one generated by the true plant $(A,B)$ and the policy $\pi$, with errors satisfying \eqref{eq:budgets} and \eqref{eq:budget-ray}. A level $\eps$ is \emph{certified} if every physical record generated at that level is certifiable; explicitly, \begin{equation}\label{eq:worst-record} \begin{aligned} &\forall\ \text{admissible physical records }\cD\quad\exists\ K_{\cD},\ P_{\cD}\succ0\\ &\forall\ (A',B')\in\Sigma_{\delta_w,\delta_v}(\cD):\\ &\quad\qquad\qquad(A'+B'K_{\cD})^\top P_{\cD}(A'+B'K_{\cD})\prec P_{\cD}. \end{aligned} \end{equation} The \emph{tolerance} $\eps_T^\kappa(\pi;A,B)$ is the supremum of all $e$ such that every level $\eps\in[0,e)$ is certified. The order matters: a different gain may be computed from each record, but all explanations of that record share its gain and Lyapunov matrix. Success on a single error-free record or a finite sample of records is not enough. The definition requires success at every smaller error level; success at the supremum itself is not claimed. \par
The principal benchmark is the \emph{causal oracle tolerance} \begin{equation}\label{eq:oracle} \Gamma(A,B)=\sup_{T\ge1}\,\sup_{\pi}\,\eps_T^0(\pi;A,B), \end{equation} where $\pi$ ranges over deterministic causal policies. The oracle knows the true plant when it chooses its policy and duration, but it must still certify every compatible alternative. Because it observes exact states, it is also a valid benchmark for measured-state experiments: a record with zero sensor error is admissible for every $\kappa$, and all its exact-state explanations remain compatible when $\kappa>0$. Hence $\eps_T^\kappa(\pi;A,B)\le\eps_T^0(\pi;A,B)\le\Gamma(A,B)$. Predetermined plant-informed experiments are special causal policies, so $\Gamma$ covers them as well. \par
\subsection{Controllability scales} Define the controllability matrices \begin{equation}\label{eq:controllability} \mathcal C_j=[B\ AB\ \cdots\ A^{j-1}B],\qquad\mathcal C=\mathcal C_n. \end{equation} For a controllable pair with minimal-polynomial degree $d$, $\smin(\mathcal C_d)>0$. In general, $\smin(\mathcal C_j)$ measures how strongly the first $j$ input-response blocks reach the weakest state direction, in the coordinates of the error bounds. We relate these singular values to the experiment tolerances $\eps_T^\kappa$ and $\Gamma$. \par
\par
\section{Information and exact minimum duration}\label{sec:experiments} \par
An input sequence is \emph{universal} if its tolerance is positive for every controllable plant of the given dimensions; the positive level may depend on the plant. We first identify the information that universality demands. \par
\begin{theorem}[Information required for positive robustness] \label{thm:information} Let $(X,Y,U)$ be an error-free record of a controllable real plant, started from the known state $x_0=0$. \begin{enumerate} \item With exact states, the record is certifiable for some positive $\delta_w$ if and only if $\operatorname{rank}X=n$. \item Regarded as a measured record, it is certifiable for some positive $\delta_w$ and $\delta_v$ if and only if \begin{equation} \operatorname{rank}\begin{bmatrix}X\\U\end{bmatrix}=n+\operatorname{rank}U. \label{eq:information-rank} \end{equation} If \eqref{eq:information-rank} fails, then for every $\delta_v>0$, even with $\delta_w=0$, no single gain stabilizes all compatible models. \end{enumerate} When the corresponding rank condition holds, for all sufficiently small error bounds, every physical record generated by the same input sequence is certifiable ($\delta_w$ in part~1; $\delta_w$ and $\delta_v$ in part~2). \end{theorem} \par
Condition \eqref{eq:information-rank} says that the state rows add $n$ independent directions to the row space of $U$. After restricting the inputs to the directions used by the experiment, this is full row rank of the reduced stacked state-input matrix. Thus an arbitrarily small sensor-error bound changes the information requirement. Exact states can determine a stabilizing gain without determining the plant. With sensor uncertainty, the error-free record must determine the state dynamics and the action of every input direction used in the experiment; the gain can ignore the unused ones. \par
\emph{Proof idea.} With exact states and $\operatorname{rank}X=n$, the zero initial state allows us to choose a right inverse $\Theta$ of $X$ for which $Y\Theta$ is Schur. The gain $K=U\Theta$ then gives $A'+B'K=Y\Theta-W'\Theta$ for every compatible model, so a strict Lyapunov inequality for $Y\Theta$ survives small process errors. Sensor errors create a different ambiguity. If $a^\top X+c^\top U=0$ with $a\ne0$, then for every $s$ sufficiently near one the rescaled trajectory $sx_t$ is consistent with the measurements and is explained, without process error, by $A'=A+va^\top$, $B'=s(B+vc^\top)$ for arbitrary $v$. A common gain would need $a^\top+sc^\top K=0$ for two distinct values of $s$, which forces $a=0$. Conversely, under \eqref{eq:information-rank} the compatible models, restricted to the used input directions, shrink to the true plant as the error bounds vanish. \par
\begin{example}[A scalar discontinuity] Take $A=2$, \mbox{$B=[1\ -2]$}, and the inputs $e_1,e_2,0$. The error-free state sequence is $(0,1,0,0)$. Exact-state alternatives satisfy $|A'+B'_2|\le\delta_w$, so $K=(0,1)^\top$ certifies every $\delta_w<1$. With any positive sensor-error bound, the underlying trajectory $(0,s,0,0)$ is consistent with the measurements for every $s$ sufficiently near one and is explained without process error by $A'=a$, $B'=[s\ -sa]$ for every $a\in\mathbb R$. For $K=(k_1,k_2)^\top$, the closed-loop scalar is $a(1-sk_2)+sk_1$. Stability for all $a$ forces $sk_2=1$; two distinct admissible values of $s$ make this impossible. \end{example} \par
The rank conditions concern one record; the minimum-duration question asks for one input sequence, fixed in advance, that satisfies them on every controllable plant. The following pulse sequences do so without resetting the state: \begin{align} U_{\mathrm{exact}}&:\ u_{jn}=e_{j+1}, & T_{\mathrm{exact}}&=mn+1, \label{eq:word-e}\\ U_{\mathrm{meas}}&:\ u_{j(n+1)}=e_{j+1}, & T_{\mathrm{meas}}&=m(n+1), \label{eq:word-s} \end{align} for $0\le j<m$; all other inputs are zero. Thus $U_{\mathrm{exact}}$ applies the unit pulses $e_1,\ldots,e_m$ every $n$ steps and $U_{\mathrm{meas}}$ every $n+1$ steps; both end with $n$ zero inputs. In $U_{\mathrm{exact}}$, the final zero input makes the last response $x_{mn}$ a column of $X$. \par
\begin{theorem}[Exact universal minimum durations]\label{thm:minimal} For real $n$-state, $m$-input plants, the minimum duration of a universal predetermined sequence is $mn+1$ with exact states and $m(n+1)$ for every fixed $\kappa>0$; $U_{\mathrm{exact}}$ and $U_{\mathrm{meas}}$ attain these minima. The lower bounds hold even for controllable diagonal plants with one eigenvalue in each of $n$ prescribed disjoint nonempty open intervals in $(0,1)$. For $m=1$, the common minimum $n+1$ is necessary even for deterministic causal policies. \end{theorem} \par
Thus any positive sensor uncertainty, however small, costs exactly $m-1$ additional transitions. \par
\emph{Proof idea.} For exact states, suppose $a^\top X=0$; for measured states, suppose $a^\top X+c^\top U=0$. The pulse spacing tests the relevant relation against $Be_j,ABe_j,\ldots,A^{n-1}Be_j$, one actuator at a time, so controllability forces $a=0$. With measured states, each pulse column also forces the corresponding entry of $c$ to vanish. For the lower bound, the samples of any shorter sequence satisfy a nonzero polynomial relation, and a diagonal plant built from it yields a nonzero state-row dependence in the exact-state case, or a relation $a^\top X+c^\top U=0$ with $a\ne0$ in the measured-state case. Its poles may be chosen in the prescribed intervals. Counting polynomial coefficients gives $m(n+1)$ for measured states; the zero initial column gives $mn+1$ for exact states. \par
Set \begin{equation}\label{eq:short-tolerances} \begin{aligned} E_{\mathrm{exact}}&=\eps_{mn+1}^0(U_{\mathrm{exact}};A,B),\\ E_{\mathrm{meas}}^\kappa&=\eps_{m(n+1)}^\kappa(U_{\mathrm{meas}};A,B). \end{aligned} \end{equation} The subscripts name the sequences, not the measurement model: $E^0_{\mathrm{meas}}$ evaluates $U_{\mathrm{meas}}$ with exact states and need not equal $E_{\mathrm{exact}}$. Universality asks only that these tolerances be positive; the next section asks how large they are compared with the best possible experiment. \par
\section{Robustness of the shortest experiments}\label{sec:spectral} The benchmark $\Gamma$ knows the plant, may choose inputs causally, and has no duration limit. Nevertheless, the shortest universal sequences retain a fixed fraction of its tolerance on a broad spectral class. Both bounds use the controllability scale $\smin(\mathcal C)$. We first give record tests connecting data errors to these bounds. \par
\subsection{Certificates from a single record}\label{sec:record-certificates} We use a single number, the record margin, both to certify noisy records and to bound what an error-free record can certify. It builds on the matrix-inequality approach to data-based stabilization~\cite{slemma2022,petersen2022,qmi2023}. For real matrices $X,Y\in\R^{n\times T}$, define $\marginrec(X,Y)\ge0$ by the semidefinite program \begin{equation} \begin{aligned} \marginrec(X,Y)^2={}&\max_{H\succeq0,\,t\in\R}t\\ \text{subject to }& \begin{bmatrix} XX^\top-H-tI&XY^\top\\ YX^\top&YY^\top+H \end{bmatrix}\succeq0. \end{aligned} \label{eq:record-margin} \end{equation} \par
\begin{lemma}[Record certificates and obstructions] \label{lem:margin} Let $X,Y,U$ be real matrices of compatible sizes; they need not come from a trajectory or start at zero. Suppose $\Sigma_{\delta_w}(X,Y,U)$ is nonempty. If $\delta_w<\marginrec(X,Y)$, a certificate for $\Sigma_{\delta_w}(X,Y,U)$ can be computed from $(X,Y,U,\delta_w)$, with no rank assumption on $[X;U]$. If some real pair fits the record exactly, that is, $Y=AX+BU$, this condition is also necessary. Moreover, for all real $E_X,E_Y$ of the sizes of $X,Y$ and all complex $q,\tau$ with $\|q\|=1$ and $|\tau|\le1$, \begin{gather} |\marginrec(X{+}E_X,Y{+}E_Y){-}\marginrec(X,Y)|\le\|E_X\|+\|E_Y\|, \label{eq:margin-lipschitz}\\ \marginrec(X,Y)\le\|q^*(X-\tau Y)\|. \label{eq:margin-pencil} \end{gather} \end{lemma} \par
Thus, for a record that some pair fits exactly (in particular, an error-free record), the certifiable process-error bounds are exactly those below $\marginrec$; the margin is measured in the same units as $\delta_w$. The Lipschitz bound \eqref{eq:margin-lipschitz} transfers a positive margin from an error-free record to nearby noisy records, which yields lower bounds on tolerance. The pencil bound \eqref{eq:margin-pencil} exhibits a direction in the data that caps what an error-free record can certify, which yields upper bounds. These are the two sides of the experimental comparison. \par
For measured states, fit the data by least squares and measure how much of the state data is independent of the inputs: \begin{equation}\label{eq:conditional-fit} \begin{gathered} Z=\begin{bmatrix}\widetilde X\\ U\end{bmatrix},\quad[\widehat A\ \widehat B]=\widetilde YZ^\dagger,\quad\widehat Y=[\widehat A\ \widehat B]Z,\\ \sigma_\perp=\sigma_{\min}\bigl(\widetilde X(I-U^\dagger U)\bigr). \end{gathered} \end{equation} Since $\sigma_{\min}$ is the least row singular value, $\sigma_\perp>0$ if and only if $\operatorname{rank}Z=n+\operatorname{rank}U$, which is \eqref{eq:information-rank} for the measured data; $\sigma_\perp$ quantifies it. \par
\begin{lemma}[Conditional measured-state enclosure]\label{lem:measured} Suppose $\Sigma_{\delta_w,\delta_v}(\mathcal D)$ is nonempty and $\delta_v<\sigma_\perp$. Then every compatible pair satisfies \begin{equation}\label{eq:measured-enclosure} \left\|\bigl([A'\ B']-[\widehat A\ \widehat B]\bigr)Z\right\| \le\eta:= \frac{\delta_w+(1+\|\widehat A\|)\delta_v} {1-\delta_v/\sigma_\perp}. \end{equation} The set of all pairs satisfying \eqref{eq:measured-enclosure} is $\Sigma_\eta(\widetilde X,\widehat Y,U)$, and it admits a certificate if and only if \begin{equation}\label{eq:measured-test} \eta<r_{\mathrm{rec}}(\widetilde X,\widehat Y). \end{equation} In that case, a certificate for the measured record can be computed from the data and the error bounds. Condition \eqref{eq:measured-test} is exact for the enclosing set but only sufficient for the measured compatibility set. \end{lemma} The enclosure measures changes in the model coefficients through their effect on the recorded data $Z$. On the used actuator subspace it is a matrix ellipsoid; unused input directions remain unconstrained. \par
Both tests use only the observations and the error bounds. For the uniform guarantees below, we first normalize the data by the data themselves; this mimics the analysis normalization of Section~\ref{sec:spectral}-B without knowledge of the plant. Choose an invertible $S$ with $SS^\top=\widetilde X\widetilde X^\top$ ($SS^\top=XX^\top$ for exact states) and set \[ \begin{aligned} \widetilde X_c&=S^{-1}\widetilde X, & \delta_{w,c}&=\|S^{-1}\|\,\delta_w,\\ \widetilde Y_c&=S^{-1}\widetilde Y, & \delta_{v,c}&=\|S^{-1}\|\,\delta_v. \end{aligned} \] For exact states, use the same formulas without tildes. For measured states, recompute the fit and $\sigma_\perp$ from the normalized data. Apply Lemma~\ref{lem:margin} (exact states) or Lemma~\ref{lem:measured} (measured states) to obtain $(K_c,P_c)$. Return $K=K_cS^{-1}$ and $P=S^{-\top}P_cS^{-1}$. If $S$ does not exist or the test fails, the procedure reports failure. Unused actuator directions need no special treatment. \par
\subsection{A uniform lower bound and comparison with the oracle} The spectral condition allows repeated and defective eigenvalues away from the unit circle; near the circle, distinct factors of the minimal polynomial must stay apart. To state this precisely, let $p_A$ be the minimal polynomial of $A$, of degree $d$, with roots repeated according to their multiplicity in $p_A$. For $|z|=1$, let $\varrho_1(z)\le\cdots\le\varrho_d(z)$ be the ordered distances from $z$ to these roots, and define the \emph{boundary margin} \begin{equation} b_{\partial}(A)=\min_{|z|=1}\prod_{j=2}^{d}\varrho_j(z), \label{eq:boundary-margin} \end{equation} where an empty product is one. In words, $b_{\partial}(A)$ is the smallest value, over the unit circle, of the product of the distances from $z$ to all roots except the nearest one. A single root may therefore reach the unit circle without making the margin vanish. Since the multiplicity of a root in $p_A$ is the size of its largest Jordan block, $b_{\partial}(A)>0$ exactly when all unit-circle eigenvalues of $A$ are semisimple. Together with a bound on the root moduli, a uniform lower bound on $b_{\partial}$ keeps two root factors from approaching the same point of the unit circle. \par
\begin{theorem}[Shortest sequences versus the causal oracle] \label{thm:competitive}\label{thm:boundary}\label{thm:short-floor} Fix $n,m\ge1$ and $M\ge0$. \begin{subequations}\label{eq:boundary-competitive} \par
\noindent\emph{(i) Uniform short-experiment guarantee.} There is a constant $\alpha>0$, depending only on $n$, $m$, and $M$, such that, for every $\kappa\ge0$, every controllable real pair with $\rho(A)\le M$ satisfies \begin{equation}\label{eq:short-guarantee} E_{\mathrm{exact}}\ge\alpha\smin(\mathcal C),\quad E_{\mathrm{meas}}^\kappa\ge\frac{\alpha\smin(\mathcal C)}{\max\{1,\kappa\}}. \end{equation} \emph{(ii) Oracle comparison.} For every $b_0>0$ there is a constant $\beta<\infty$, depending only on $n$, $M$, and $b_0$, such that every such pair with $b_{\partial}(A)\ge b_0$ also satisfies \begin{equation}\label{eq:boundary-oracle} \Gamma(A,B)\le\sqrt m\,\beta\smin(\mathcal C). \end{equation} \end{subequations} The lower bounds in part~(i) are guaranteed by the data-based certificate procedure of Section~\ref{sec:record-certificates}, which uses only the record and the stated error bounds. \end{theorem} \par
Combining both parts under the conditions of part~(ii) gives \begin{equation*} \frac{\alpha\smin(\mathcal C)}{\max\{1,\kappa\}}\le E_{\mathrm{meas}}^\kappa\le\Gamma(A,B)\le\sqrt m\,\beta\smin(\mathcal C). \end{equation*} Thus the measured pulse sequence retains at least the fraction $\alpha/(\max\{1,\kappa\}\sqrt m\,\beta)$ of $\Gamma$; for exact states, drop the factor $\max\{1,\kappa\}$. These constants require neither a bound on $\|A\|$ nor a bound on the conditioning of its eigenvectors, and arbitrary stable and unstable Jordan structure is allowed away from the unit circle. The spectral condition is sufficient but not necessary; Section~\ref{sec:examples}-A gives competitive families that violate it. \par
\emph{Proof idea.} Factor the controllability Gramian as $LL^\top=\mathcal C\mathcal C^\top$ and use the analysis coordinates $\bar x=L^{-1}x$. Then \[ \begin{gathered} \overline{\mathcal C}\,\overline{\mathcal C}^{\top}=I, \quad\|L^{-1}\|=1/\smin(\mathcal C),\\ \|L^{-1}W\|\le\delta_w/\smin(\mathcal C), \quad\|L^{-1}V\|\le\delta_v/\smin(\mathcal C), \end{gathered} \] where $\overline{\mathcal C}=L^{-1}\mathcal C$. Cayley--Hamilton bounds the normalized matrices using only $n$ and $M$. Together with the displayed identity, this places them in a compact family of controllable pairs. The error-free pulse records then have uniformly positive certification margins. The weak direction enters only through the error amplification $1/\smin(\mathcal C)$. This normalization depends on the unknown plant and is only for analysis; the procedure uses the observed data ($S$ above). \par
For the oracle bound, consider the error-free run of an arbitrary causal policy at a certified level; its record must still be certifiable. Polynomial (Chinese-remainder) identities split the roots of $p_A$ into three groups: inside, outside, and near the unit circle. Filtered responses bound the stable group, robustness to actuator perturbations bounds the unstable group, and the pencil bound \eqref{eq:margin-pencil} bounds each isolated root near the circle. Uniformly bounded polynomial coefficients then recombine these estimates along the least controllable direction. At no point is the norm of an eigenvector matrix or spectral projector estimated. The bound holds for every duration and policy, hence for $\Gamma$. Weak controllability can still reduce the absolute tolerance, but the theorem shows that it limits the oracle on the same scale. \par
\subsection{Stable plants and the stability gap} For stable plants with $\rho(A)\le r<1$, the upper bound depends explicitly on the stability gap $1-r$ and the minimal-polynomial degree $d$. \par
\begin{corollary}[Stable-disk comparison]\label{thm:stable} If $(A,B)$ is controllable and $\rho(A)\le r<1$, then \begin{equation} \Gamma(A,B)\le\sqrt{m\tbinom{2d-2}{d-1}}\,\frac{\smin(\mathcal C_d)}{(1-r)^{d-1}}. \label{eq:stable-ceiling} \end{equation} \end{corollary} \par
No root separation or diagonalizability is assumed. Combined with Theorem~\ref{thm:competitive}(i) for $M=1$, the shortest measured sequence retains at least the fraction \[ \frac{\alpha(n,m,1)(1-r)^{d-1}} {\max\{1,\kappa\}\sqrt{m\tbinom{2d-2}{d-1}}} \] of $\Gamma$; for exact states, drop the factor $\max\{1,\kappa\}$. The proof writes the input-to-state response through the scalar recurrence defined by $p_A$ and compares consecutive states so as to cancel one factor of $p_A$ on the unit circle; the remaining $d-1$ factors produce the exponent. Section~\ref{sec:duration} shows that this exponent cannot be reduced, even for semisimple plants. The constant multiplying the gap power is not claimed to be optimal. \par
\section{The duration needed for competitiveness}\label{sec:duration} Let $(G,B)$ be a fixed controllable real pair, and consider the slow plants $A_h=I+hG$ as $h\downarrow0$. The \emph{controllability depth} of $(G,B)$ is \begin{equation} \nu=\min\{j:\operatorname{rank}[B\ GB\ \cdots\ G^{j-1}B]=n\}. \label{eq:depth} \end{equation} For $h>0$, $G$ and $A_h$ have the same minimal-polynomial degree $d$. The degree $d$ is the order of a scalar recurrence, whereas $\nu$ counts the input-response blocks needed to span the state space; in general $\nu\le d$, with equality for single-input pairs. If $\nu\ge2$, the input reaches some state directions only through successive powers of $G$. These directions remain reachable for every $h>0$, but their responses over any fixed number of steps are weak; they become visible on the time scale $1/h$. \par
\begin{theorem}[Necessary and sufficient duration order] \label{thm:duration} Fix $(G,B)$ with $\nu\ge2$ and fix $\kappa\ge0$. As $h\downarrow0$, with all quantities evaluated at $(A_h,B)$, \begin{equation}\label{eq:slow-scales} \smin(\mathcal C)\asymp E_{\mathrm{exact}}\asymp E_{\mathrm{meas}}^{\kappa}\asymp h^{\nu-1}, \quad\Gamma\asymp1. \end{equation} There are $\beta,h_0>0$ such that, for all $h\in(0,h_0]$, all $T\ge1$, and every deterministic causal policy $\pi$, \begin{equation} \eps_T^\kappa(\pi;A_h,B)\le\beta\min\{1,(hT)^{\nu-1}\}. \label{eq:slow-causal} \end{equation} Conversely, the predetermined sequence that holds $e_1,\ldots,e_m$ in turn, each for $N=\lceil1/h\rceil$ steps, has duration $mN$ and, for all sufficiently small $h$, a tolerance bounded below by a positive constant. \end{theorem} \par
Every experiment attaining a fixed positive fraction of $\Gamma$ therefore needs duration $\Omega(1/h)$, and a sequence of duration $O(1/h)$ attains some fixed positive fraction. The result covers stable, unstable, semisimple, and defective generators. The attaining sequence must know the time scale $1/h$ to choose its block length; learning an unknown time scale is a separate question. The conclusion identifies the required duration scale, without specifying the entire optimal tolerance-versus-duration curve. \par
The upper bound \eqref{eq:slow-causal} comes from a second plant that produces the same record and cannot be stabilized. Choose a real unit vector $q$ with $q^\top[B\ GB\ \cdots\ G^{\nu-2}B]=0$, and let the true process disturbance be $w_t=-hqq^\top Gx_t$, with zero sensor error. The resulting trajectory is also an error-free trajectory of \begin{equation}\label{eq:slow-alternative} A'=I+h(I-qq^\top)G,\quad B'=B, \end{equation} which satisfies $q^\top A'=q^\top$ and $q^\top B'=0$. This alternative has an uncontrollable eigenvalue at one, so no gain can stabilize it. Since $w_t$ depends only on the current state, the construction applies along the run of every causal policy. The depth condition makes the normalized disturbance size $\|W\|/\sqrt T$ of order at most $(hT)^{\nu-1}$ when $hT\le1$ (\appref{app:duration}). Thus the necessary duration is not a limitation of a particular certificate. \par
For the matching lower bound, hold each input direction for about $1/h$ steps. For error-free data, $hx_t$ approaches a trajectory of $\dot y=Gy+Bu$, with positive definite state-input Gramian by controllability. Small error levels perturb this approximation uniformly over all physical records and small $h$. Every compatible rescaled pair $((A'-I)/h,B')$ therefore stays close to $(G,B)$. One continuous-time certificate then yields a common discrete-time certificate and positive tolerance uniformly for small $h$. \par
\emph{Sharpness of the stable-gap exponent.} For $n\ge2$, take a controllable single-input pair $(G,B)$ with $G$ diagonal and distinct negative eigenvalues. Then $d=\nu=n$ and $1-\rho(I+hG)\asymp h$, while Theorem~\ref{thm:duration} gives $\smin(\mathcal C_d)\asymp h^{d-1}$ and $\Gamma\asymp1$. Replacing the exponent $d-1$ in \eqref{eq:stable-ceiling} by any smaller one would make its right-hand side tend to zero as $h\downarrow0$, contradicting $\Gamma\asymp1$. Thus $d-1$ is sharp already for semisimple plants. \par
\section{Examples}\label{sec:examples} \subsection{Weak links and slow actuation} A small controllability scale $\smin(\mathcal C)$ can mean that every experiment is sensitive to noise, or that a short experiment misses a slower response. Fix an error ratio $\kappa\ge0$, and consider the following two-state, single-input families as $\mu,h\downarrow0$, with $0<\mu\le1$: \begin{alignat*}{3} &\text{Weak link:}\quad&A_\mu&=\begin{bmatrix}0&\mu\\0&0\end{bmatrix}, &\quad B&=\begin{bmatrix}0\\1\end{bmatrix},\\[1mm] &\text{Slow actuation:}\quad&A_h&=\begin{bmatrix}1-h&0\\0&1-2h\end{bmatrix}, &\quad B&=\begin{bmatrix}1\\1\end{bmatrix}. \end{alignat*} In both cases $m=1$ and $n=2$, so $U_{\mathrm{exact}}=U_{\mathrm{meas}}=(1,0,0)$, and the relevant tolerance is $E^\kappa_{\mathrm{meas}}$. The comparison with the oracle is: \begin{center} \begin{tabular}{@{}lccc@{}} \toprule Family & $E_{\mathrm{meas}}^\kappa$ & $\Gamma$ & $E_{\mathrm{meas}}^\kappa/\Gamma$\\ \midrule Weak link & $\asymp\mu$ & $\asymp\mu$ & $\asymp1$\\ Slow actuation & $\asymp h$ & $\asymp1$ & $\asymp h$\\ \bottomrule\end{tabular} \end{center} For the weak link, $\smin(\mathcal C)=\mu$, while the minimal polynomial is $z^2$ and $b_{\partial}(A_\mu)=1$. Theorem~\ref{thm:competitive} therefore places $E^\kappa_{\mathrm{meas}}$ and $\Gamma$ on the same scale $\mu$: longer experiments cannot improve the order of the tolerance, because the weak link limits even an experiment designed with complete knowledge of the plant. \par
The slow-actuation family has generator $G=\diag(-1,-2)$ and depth $\nu=2$. By Theorem~\ref{thm:duration}, $E^\kappa_{\mathrm{meas}}/\Gamma\asymp h\to0$, and order $1/h$ transitions are necessary, and suffice, to recover a fraction of $\Gamma$ bounded away from zero. Both families have weak finite controllability and the same shortest universal experiment, but only the second gains an unbounded factor from a longer experiment. This distinguishes absolute robustness from the cost of using the shortest possible experiment. \par
Actuation also explains why the spectral condition is only sufficient. With the same $A_h$ but two direct inputs, $B=I_2$, one has $\smin(\mathcal C)\ge1$ although \mbox{$b_{\partial}(A_h)=2h\to0$}. The lower bound of Theorem~\ref{thm:competitive}(i) stays positive, while cancelling the input by $w_t=-Bu_t$ shows $\Gamma\le\sqrt2$ (\appref{app:duration}). These fully actuated plants therefore remain uniformly competitive despite the crowded spectrum. \par
\subsection{A nonnormal oscillator} Consider \[ A_k=\begin{bmatrix}0&-k\\ k^{-1}&0\end{bmatrix},\quad B=\begin{bmatrix}1\\0\end{bmatrix},\quad k\ge1. \] Although $\|A_k\|=k\to\infty$, its eigenvalues remain $\pm i$. Indeed, $A_k^2=-I$, $b_{\partial}(A_k)=\sqrt2$, and $\mathcal C_2=\operatorname{diag}(1,k^{-1})$, so $\smin(\mathcal C)=k^{-1}$. For every fixed error ratio $\kappa$, \[ E_{\mathrm{exact}}(A_k,e_1) \asymp E_{\mathrm{meas}}^\kappa(A_k,e_1) \asymp\Gamma(A_k,e_1)\asymp k^{-1}. \] The shortest sequence $(1,0,0)$ therefore remains uniformly competitive despite $\|A_k\|\to\infty$. For a direct lower bound, the analysis factor $L_k=\diag(1,k^{-1})$ satisfies $L_k^{-1}A_kL_k=A_1$ and $L_k^{-1}B=B$. The fixed rotation pair $(A_1,e_1)$ has positive tolerance for the pulse sequence, while the transformed error bounds grow by at most $k$. Conversely, the disturbance $w_t=(0,-x_{1,t}/k)^\top$ keeps $x_{2,t}=0$ and $x_{1,t+1}=u_t$. Its normalized size $\|W\|/\sqrt T$ is at most $k^{-1}$, and the record also fits the unstabilizable pair \[ A'_k=\begin{bmatrix}0&-k\\0&1\end{bmatrix},\quad B'=B=\begin{bmatrix}1\\0\end{bmatrix}. \] For this family, increasing $k$ reduces the shortest-experiment tolerance and the oracle tolerance at the same rate. \par
\section{Conclusions}\label{sec:conclusions} Minimal universal duration and competitive duration answer different questions. Sensor uncertainty changes the information requirement and the exact shortest predetermined experiment. Once that information is present, bounded spectral radius and the stated root separation make the shortest sequences competitive with plant-informed causal experimentation of arbitrary fixed duration. Slow actuation can defeat this comparison, even for diagonal dynamics, and then the controllability depth determines the loss of the shortest sequences while $1/h$ determines the required duration scale. The results leave open adaptive multi-input minimum lengths, sharper competitive constants, and experiments that adapt their duration to an unknown time scale. \par
\appendix[Proofs] \label{app:proofs} Theorem~\ref{thm:competitive} is proved in two stages: the uniform lower bound in \appref{app:short-floor} and the oracle upper bound in \appref{app:boundary}. The intervening stable-disk proof in \appref{app:stable} develops the transfer identity used in the oracle argument. \par
\subsection{Proof of Theorem~\ref{thm:information}}\label{app:information} If $a^\top X=0$ for some $a\ne0$, all plants $(A+tva^\top,B)$ fit the exact record. For a fixed gain $K$, choose $v$ with $a^\top v\ne0$. Their closed-loop traces are unbounded as $t$ varies, whereas a Schur $n\times n$ matrix has trace of modulus less than $n$. This proves exact-state necessity. \par
Conversely, assume $X$ has full row rank and put $\Pi=I-X^\dagger X$. The identity $Y=(YX^\dagger)X+Y\Pi$ realizes the given zero-start sequence as a trajectory of $(YX^\dagger,Y\Pi)$ with input $e_{t+1}\in\R^T$ at time $t$, the corresponding standard basis vector. Its reachable subspace contains every column of $X$, so this pair is controllable. Choose $K_0$ so that $Y\Theta$ is Schur, where $\Theta=X^\dagger+\Pi K_0$. Then $X\Theta=I$, and $K=U\Theta$ gives \[ A'+B'K=Y\Theta-W'\Theta\] for every compatible alternative. A strict Lyapunov inequality for $Y\Theta$ persists for all sufficiently small $\|W'\|$, proving the first equivalence. \par
If \eqref{eq:information-rank} fails, a left-null vector of the regressor exists with nonzero state component: $a^\top X+c^\top U=0$, $a\ne0$. Fix any $\delta_v>0$. For all $s$ in a sufficiently small open interval containing one, the underlying trajectory $x'_t=sx_t$ has sensor error $(1-s)[x_0\ \cdots\ x_T]$ of norm less than $\delta_v$. For every $v\in\mathbb R^n$, it is explained with zero process error by \[ A'=A+va^\top,\quad B'=s(B+vc^\top). \] A common gain $K$ would require $a^\top+s c^\top K=0$ for every such $s$, since otherwise the closed-loop trace is unbounded in $v$. Two distinct values of $s$ force $c^\top K=0$ and $a=0$, contradicting the chosen nonzero state component. \par
Finally, suppose \eqref{eq:information-rank} holds. Choose orthonormal columns $Q$ spanning $\operatorname{im}U$ and write $U=Q U_r$. The reduced regressor $[X;U_r]$ has full row rank. All columns of $X$ are reachable under $(A,BQ)$, so this pair is controllable. Lemma~\ref{lem:measured} places the reduced compatible pairs in a neighborhood shrinking to $(A,BQ)$ as both error bounds tend to zero. A certificate with a stabilizing reduced gain $K_r$ persists on this neighborhood; use $K=QK_r$ in the original input coordinates. For a fixed finite input sequence, the physical and measured data blocks depend continuously on the error arrays. The required ranks, controllability, and strict Lyapunov inequalities therefore persist for every sufficiently small admissible perturbation. This proves local robustness for all records in one neighborhood. \qed\par
\subsection{Proof of Theorem~\ref{thm:minimal}}\label{app:experiments} We first verify the two pulse sequences. For the measured sequence, a relation $a^\top X+c^\top U=0$ gives $c_1=0$ at $t=0$, followed by $a^\top A^kBe_1=0$, $0\le k<n$, on the next $n$ columns. Cayley--Hamilton extends these equalities to every $k$. The response to the first pulse is therefore annihilated at all later times. The next pulse column gives $c_2=0$, and its following columns treat $Be_2$. Induction and controllability yield $a=c=0$. For the exact sequence, a relation $a^\top X=0$ gives the same $n$ responses for each actuator. The state at the next pulse time is recorded before that pulse acts, so it supplies the preceding actuator's last required response. For the last actuator, this response is in the final current-state column $x_{mn}$. Thus controllability gives $a=0$. The local robustness assertion in Theorem~\ref{thm:information} now gives positive tolerance for each plant. \par
For the lower bounds, we realize a polynomial relation by a controllable stable diagonal plant. Let $p(z)\in\R[z]^{1\times m}$ be nonzero with $\deg p\le n$. Choose distinct nonzero $\lambda_1,\ldots,\lambda_n$ in the prescribed disjoint nonempty open intervals, avoiding the finitely many common zeros of the components of $p(1/\lambda)$. Set \[ A=\diag(\lambda_1,\ldots,\lambda_n),\quad B_i=\frac{\lambda_i p(1/\lambda_i)} {\prod_{j\ne i}(1-\lambda_j/\lambda_i)}, \] where $B_i$ is row $i$ of $B$. With $q_A(z)=\prod_{i=1}^n(1-\lambda_i z)$, this gives \[ \frac{p(z)}{q_A(z)} =p(0)+z\mathbf1_n^\top(I-zA)^{-1}B. \] The identity follows by evaluating the degree-$n$ numerators at $0,1/\lambda_1,\ldots,1/\lambda_n$. Each row of $B$ is nonzero; since the eigenvalues are distinct, the pair is controllable by PBH. \par
Fix an input sequence and put $u(z)=\sum_{t<T}u_tz^t$. If $T<m(n+1)$, the $m(n+1)$ coefficients of $p$ can solve the $T$ homogeneous equations $p(z)u(z)=0\pmod{z^T}$ nontrivially. The zero-start state series is $x(z)=z(I-zA)^{-1}Bu(z)$. Since $q_A(0)=1$, division by $q_A$ preserves the first $T$ vanishing coefficients. The realization therefore gives $\mathbf1_n^\top x_t+p(0)u_t=0$ for $t<T$. This relation has a nonzero state row, so, by Theorem~\ref{thm:information}(2), the error-free record is not certifiable for any $\delta_v>0$. \par
For exact states, if $T<mn+1$, choose a nonzero row polynomial $q$ of degree at most $n-1$ such that $q(z)u(z)$ has its first $T-1$ coefficients zero. Such a polynomial exists because there are $mn>T-1$ unknown coefficients. Taking $p=zq$ in the same realization gives $\mathbf1_n^\top X=0$. By the necessity proof of Theorem~\ref{thm:information}(1), no gain stabilizes all exact-state explanations, even without process errors. Finally, any causal run with $T\le n$ has $\rank X\le T-1<n$, because $x_0=0$. For one input, the pulse sequence attains the resulting lower bound. \qed\par
\subsection{Proof of Lemma~\ref{lem:margin}}\label{app:records} The SDP is feasible at $H=0,t=0$, and its nonnegative feasible levels and matrices are bounded by the top-left block. Its maximum is therefore finite and attained. If $\delta_w<\marginrec(X,Y)$, an optimizer $H_*$ gives the positive definite matrix \[ H=H_*+\tfrac12\bigl(\marginrec(X,Y)^2-\delta_w^2\bigr)I. \] The block in \eqref{eq:record-margin} is then strictly positive at $t=\delta_w^2$. Put \[ \begin{gathered} Q_H=I+Y^\top H^{-1}Y,\quad R_H=XQ_H^{-1}X^\top,\\ \Theta=Q_H^{-1}X^\top R_H^{-1}. \end{gathered} \] With $D=Y\Theta$, Woodbury's identity and a Schur complement give \mbox{$R_H\succ H+\delta_w^2I$}, $X\Theta=I$, and \[ \Theta^\top\Theta+D^\top H^{-1}D =R_H^{-1}\prec(H+\delta_w^2I)^{-1}. \] For every compatible residual $W'$ one has \[ \big\|(H+\delta_w^2I)^{-1/2}[H^{1/2}\ -W']\big\|\le1. \] Consequently $K=U\Theta$ and $P=(H+\delta_w^2I)^{-1}$ form a certificate, since $A'+B'K=D-W'\Theta$. This argument requires no rank assumption on $[X;U]$ and proves sufficiency for every nonempty noisy record. \par
Now assume $Y=AX+BU$, and let $Z=[X;U]$ and $\Pi=Z^\dagger Z$. Any common gain must have $[I;K]\in\operatorname{im}Z$: otherwise a rank-one perturbation of $[A\ B]$ that vanishes on $\operatorname{im}Z$ makes the closed-loop trace unbounded. Hence it has a representation $K=U\Theta$, $X\Theta=I$, $\Pi\Theta=\Theta$. Its compatible closed-loop family is precisely $(Y-\delta_w\Delta)\Theta$, $\|\Delta\|\le1$. For $\delta_w>0$, strict full-block elimination~\cite[Fact~1]{petersen2022}, with its positive multiplier absorbed into the Lyapunov matrix, makes existence of a certificate equivalent to some $H\succ0$ and such a right inverse satisfying \[ \Theta^\top(I+Y^\top H^{-1}Y)\Theta\prec(H+\delta_w^2I)^{-1}. \] For $\delta_w=0$, the same condition follows by scaling a Lyapunov matrix $P$ for the Schur matrix $Y\Theta$ until $P-(Y\Theta)^\top P(Y\Theta)\succ\Theta^\top\Theta$, and setting $H=P^{-1}$. For fixed $H$, weighted least squares minimizes the left side in the Loewner order at $R_H^{-1}$. Its minimizer lies in $\operatorname{row}Z$, since the rows of $X,Y$ lie there and $Q_H$ preserves that subspace. The resulting condition \mbox{$R_H\succ H+\delta_w^2I$} is exactly the strict margin LMI. A strict block at $t=\delta_w^2$ remains feasible for a slightly larger $t$. Thus the record is certifiable if and only if $\delta_w<\marginrec(X,Y)$, including failure at equality. \par
For the remaining properties, SDP duality gives \[ \marginrec(X,Y)^2= \min_{\Omega\in\mathcal S_n} \tr\!\left(\Omega\begin{bmatrix}X\\Y\end{bmatrix} \begin{bmatrix}X\\Y\end{bmatrix}^{\!\top}\right), \] where \[ \mathcal S_n=\left\{ \Omega=\begin{bmatrix}\Omega_{11}&\Omega_{12}\\\Omega_{12}^\top&\Omega_{22}\end{bmatrix}\succeq0: \tr\Omega_{11}=1,\ \Omega_{22}\preceq\Omega_{11}\right\}. \] This set is compact; $\Omega_{11}=I/n$, $\Omega_{22}=I/(2n)$, $\Omega_{12}=0$ is strictly dual feasible, so strong duality applies. Factoring $\Omega=[F_1;F_2][F_1;F_2]^\top$ gives $\|F_1\|_F=1$, $\|F_2\|_F\le1$; the square root of its objective is $\|F_1^\top X+F_2^\top Y\|_F$. The triangle inequality, followed by minimization, proves \eqref{eq:margin-lipschitz}. Finally the real part of the Hermitian rank-one matrix generated by $[q;-\overline\tau q]$ belongs to $\mathcal S_n$ and has objective $\|q^*(X-\tau Y)\|^2$. This proves \eqref{eq:margin-pencil}.\qed\par
\subsection{Proof of Lemma~\ref{lem:measured}}\label{app:measured} For a compatible pair, define \mbox{$\Delta=[A'\ B']-[\widehat A\ \widehat B]$} and put $P_U=I-U^\dagger U$. Its data satisfy \[ \widetilde Y-V'_+=A'(\widetilde X-V'_-)+B'U+W'. \] Projecting this equation onto the row space of $Z$ and using $\widehat Y=\widetilde YZ^\dagger Z$ gives \[ \Delta Z=(A'V'_- -W'-V'_+)Z^\dagger Z. \] Since $UP_U=0$ and $\widetilde XP_U$ has full row rank, \[ A'-\widehat A=\Delta ZP_U(\widetilde XP_U)^\dagger. \] Thus $\|A'\|\le\|\widehat A\|+\|\Delta Z\|/\sigma_\perp$. Both sensor blocks inherit the norm bound on their shared array $V'$. As $Z^\dagger Z$ is an orthogonal projector, these identities imply \[ \|\Delta Z\| \le\delta_w+(1+\|\widehat A\|)\delta_v +\frac{\delta_v}{\sigma_\perp}\|\Delta Z\|. \] Since $\delta_v<\sigma_\perp$, rearrangement proves \eqref{eq:measured-enclosure}. Its enclosing set is $\Sigma_\eta(\widetilde X,\widehat Y,U)$, which has the zero-residual explanation \mbox{$\widehat Y=\widehat A\widetilde X+\widehat B U$}. The zero-residual part of Lemma~\ref{lem:margin} therefore gives the equivalence \eqref{eq:measured-test} and a certificate. That lemma is algebraic: the fitted blocks need not themselves form a shifted state trajectory. Containment gives the claimed certificate for the original measured record.\qed\par
\subsection{Proof of Theorem~\ref{thm:competitive}: the short-experiment lower bound} \label{app:short-floor} Choose $LL^\top=\mathcal C\mathcal C^\top$ and set $\bar A=L^{-1}AL$, $\bar B=L^{-1}B$, and $\overline{\mathcal C}=L^{-1}\mathcal C$. These analysis coordinates are used only in the proof; the experimenter instead computes data-normalized coordinates from the observed record. Then $\|L^{-1}\|=1/\smin(\mathcal C)$ and $\overline{\mathcal C}\overline{\mathcal C}^\top=I$. Cayley--Hamilton gives \[ \begin{gathered} \bar A=\overline{\mathcal C}(J_A\otimes I_m) \overline{\mathcal C}^\top,\\ \|\bar A\|\le c_{n,M}:=\sqrt{n-1}+(1+M)^n, \quad\|\bar B\|\le1, \end{gathered} \] where $J_A$ is the companion matrix of the characteristic polynomial of $A$. The bound follows from the elementary symmetric functions of roots of modulus at most $M$. These bounds and the identity $\mathcal C(\bar A,\bar B)\mathcal C(\bar A,\bar B)^\top=I$ define a compact family consisting entirely of controllable pairs. \par
Put $\zeta=\max\{\delta_w,\delta_v\}/\smin(\mathcal C)$, taking $\delta_v=0$ for exact states, and let $\bar X^0$ be the error-free current-state block generated by $(\bar A,\bar B)$ and the pulse sequence at hand. At either fixed duration, summing the propagated errors bounds state perturbations in analysis coordinates by $O(\zeta)$ uniformly over this family; sensor perturbations add at most $\zeta$. For example, the current-state bound is $\zeta\sum_{j=0}^{T-2}c_{n,M}^j$, using temporal shifts of norm one. The error-free pulse ranks and compactness give uniform positive least singular values for $\bar X^0$, and for $[\bar X^0;U]$ in the measured case. These ranks persist for sufficiently small $\zeta$. \par
Let $S$ be the data factor of Section~\ref{sec:record-certificates} and put $R=L^{-1}S$. Since $RR^\top=(L^{-1}\widetilde X)(L^{-1}\widetilde X)^\top$, both $\|R\|$ and $\|R^{-1}\|$ are uniformly bounded. Thus the data-normalized coordinates $x_c=S^{-1}x$ differ from the analysis coordinates by a uniformly bounded transformation. This argument allows every admissible factor $S$ and requires no continuous choice of factor. Indeed, any sequence with $\zeta\to0$ has a subsequence on which $(\bar A,\bar B,R)$ converges. Its limit satisfies $RR^\top=\bar X^0(\bar X^0)^\top$, so the data-normalized records approach the compact family of error-free data-normalized records uniformly. \par
For exact states, those error-free records have positive margins by Theorem~\ref{thm:information}(1) and the exact part of Lemma~\ref{lem:margin}. The Lipschitz bound~\eqref{eq:margin-lipschitz} therefore bounds all sufficiently small perturbed-record margins uniformly away from zero. Meanwhile, \[ \delta_{w,c}=\|S^{-1}\|\delta_w \le\|R^{-1}\|\zeta=O(\zeta). \] The sufficient test $\delta_{w,c}<r_{\mathrm{rec}}(X_c,Y_c)$ consequently holds for every physical record when $\zeta$ is sufficiently small. This argument does not require full row rank of $[X_c;U]$. \par
For measured states, the error-free data-normalized regressors have full row rank with a uniform positive least singular value. This remains true for sufficiently small $\zeta$. Apply the fit \eqref{eq:conditional-fit} in the data-normalized coordinates. Least squares is continuous on this neighborhood, so the fitted records $(\widetilde X_c,\widehat Y_c)$ approach the error-free data-normalized records uniformly, and $\|\widehat A_c\|$ stays bounded. The error-free records have positive margins by Theorem~\ref{thm:information}(2) and Lemma~\ref{lem:margin}; continuity therefore gives a uniform positive lower bound on $r_{\mathrm{rec}}(\widetilde X_c,\widehat Y_c)$. The same compactness argument bounds $\sigma_\perp$, computed in the data-normalized coordinates, away from zero. Both transformed error bounds are $O(\zeta)$, so Lemma~\ref{lem:measured} gives an enclosure radius $O(\zeta)$. Its sufficient test consequently certifies every physical record when $\zeta$ is sufficiently small. \par
Both procedures use only observations and the stated bounds, and their certificates pull back to physical coordinates. Under \eqref{eq:budget-ray}, $\zeta=\eps\sqrt{T_{\mathrm{exact}}}/\smin(\mathcal C)$ for exact states and $\zeta=\max\{1,\kappa\}\eps\sqrt{T_{\mathrm{meas}}}/\smin(\mathcal C)$ for measured states. Absorbing these fixed durations into a positive constant \mbox{$\alpha(n,m,M)$} proves the lower bounds. \qed\par
\subsection{Proof of Corollary~\ref{thm:stable}}\label{app:stable} We prove a slightly sharper bound. Define \begin{equation}\label{eq:stable-coefficient} \begin{gathered} a_{d,k}(r)=\sum_{j=k}^{d-1}\binom jk r^{j-k}(1-r)^{d-1-j},\\ R_d(r)=\Bigl(\sum\nolimits_{k=0}^{d-1}a_{d,k}(r)^2\Bigr)^{1/2}. \end{gathered} \end{equation} The proof establishes both bounds in the chain \begin{equation}\label{eq:stable-sharper} \begin{split} \Gamma(A,B)&\le\frac{\sqrt m\,R_d(r)}{(1-r)^{d-1}}\smin(\mathcal C_d)\\ &\le\frac{\sqrt{m\binom{2d-2}{d-1}}}{(1-r)^{d-1}}\smin(\mathcal C_d). \end{split} \end{equation} \par
\par
For a monic polynomial $p(z)=\sum_{j=0}^\ell c_jz^j$, define \[ \phi_{p,k}(z)=\sum_{j=k+1}^\ell c_jz^{j-1-k},\quad\phi_p=(\phi_{p,0},\ldots,\phi_{p,\ell-1})^\top. \] For $p=p_A$, abbreviate these to $\phi_k$ and $\phi$. Polynomial division and $p_A(A)=0$ give \begin{equation} (zI-A)^{-1}B=\frac{\sum_{k=0}^{d-1}\phi_k(z)A^kB}{p_A(z)}. \label{eq:app-transfer} \end{equation} Choose a real unit vector $q$ with \mbox{$\|q^\top\mathcal C_d\|=\smin(\mathcal C_d)$}. On the error-free trajectory of any causal policy, $\|U\|_F\le\sqrt{mT}$. Extend this input by zeros. Stability and Parseval's identity bound the norm of the first $T$ samples of $q^\top(x_t-\tau x_{t+1})$ by \[ \sqrt{mT}\,\smin(\mathcal C_d)\max_{|z|=1} \frac{|1-\tau z|\,\|\phi(z)\|}{|p_A(z)|}. \] By \eqref{eq:margin-pencil}, every certified level is bounded by this expression divided by $\sqrt T$. This obstruction applies for complex $|\tau|\le1$, and zero sensor error makes it an obstruction for every $\kappa$. \par
Write $p_A(z)=\prod_{j=1}^d(z-\lambda_j)$ and choose $\tau=\overline\lambda_1$. The vector $\phi(z)$ contains the coefficients in $w$ of the divided difference in the telescoping identity \[ \frac{p_A(z)-p_A(w)}{z-w} =\sum_{j=1}^d\prod_{\ell<j}(w-\lambda_\ell) \prod_{\ell>j}(z-\lambda_\ell). \] Cancelling $|1-\overline\lambda_1z|=|z-\lambda_1|$ on $|z|=1$ reduces the bound to the coefficient norm of \[ 1+\sum_{j=2}^d \frac{\prod_{\ell<j}(w-\lambda_\ell)} {\prod_{\ell=2}^{j}(z-\lambda_\ell)}. \] Its $w^k$ coefficient has modulus at most \[ \sum_{j=k}^{d-1}\binom jk r^{j-k}(1-r)^{-j} =\frac{a_{d,k}(r)}{(1-r)^{d-1}}, \] with $a_{d,k}(r)$ as in \eqref{eq:stable-coefficient}. Taking the Euclidean coefficient norm proves the $R_d$ bound. For the binomial estimate, \[ a_{d,k}(r)\le\binom{d-1}{k} \sum_{j=k}^{d-1}r^{j-k}(1-r)^{d-1-j} \le\binom{d-1}{k}. \] The last sum is bounded by $(r+1-r)^{d-1-k}$ after inserting its binomial coefficients. Hence $R_d^2\le\sum_k\binom{d-1}{k}^2=\binom{2d-2}{d-1}$. Taking suprema over certified levels, durations, and causal policies proves the result. Since $\smin(\mathcal C_d)\le\smin(\mathcal C)$, the competitive fraction follows. \qed\par
\par
\par
\subsection{Proof of Theorem~\ref{thm:competitive}: the oracle ceiling} \label{app:boundary}\label{app:spectral} The lower bounds of Theorem~\ref{thm:competitive}(i) were proved in \appref{app:short-floor}. Fix a causal policy and a certified level $\eps>0$, and use its error-free record. Its exact-state compatibility set must admit a certificate. \par
\emph{Modal bound.} First let $v^*$ be a left eigenvector, $v^*A=\lambda v^*$, and set \[ s=v^*X,\quad y=v^*Y,\quad f=v^*BU. \] Then $y=\lambda s+f$. Since $x_0=0$, $s$ is $y$ delayed by one step, so $\|s\|\le\|y\|$. If $|\lambda|\le1$, choose $\tau=\overline\lambda$ and use \[ \|s-\overline\lambda y\|^2 +(1-|\lambda|^2)(\|y\|^2-\|s\|^2)=\|f\|^2. \] If $|\lambda|\ge1$, choose $\tau=\lambda^{-1}$, so that $s-\tau y=-f/\lambda$. Both choices satisfy $|\tau|\le1$. Since $\|f\|\le\sqrt{mT}\|v^*B\|$, the record obstruction gives \begin{equation} \eps\|v\|\le\frac{\sqrt m\,\|v^*B\|}{\max\{1,|\lambda|\}}. \label{eq:app-modal} \end{equation} If $d=1$, then $A$ is a scalar multiple of $I$, so every row is modal. Applying this estimate to a unit vector $v$ with \mbox{$\|v^*B\|=\smin(B)$} proves the upper bound, since $\smin(B)\le\smin(\mathcal C)$. \par
\emph{Spectral grouping.} For $d\ge2$, set \[ \theta=\frac14\min\{1,b_0/(1+M)^{n-2}\}. \] Every root with $||\lambda|-1|\le\theta$ is simple in $p_A$ and is separated from all other roots by at least $2\theta$. Otherwise, at $z=\lambda/|\lambda|$ the second-smallest root distance is less than $3\theta$, and \[ \prod_{j=2}^{d}\varrho_j(z) <3\theta(1+M)^{d-2}\le3b_0/4, \] contrary to the boundary-margin assumption. \par
Factor $p_A=p_s p_u\prod_i(z-\lambda_i)$, grouping roots with modulus below $1-\theta$, above $1+\theta$, and the remaining simple roots, respectively. Omit empty factors. Different groups are separated by at least $2\theta$. For each factor $p_j$, let $\psi_j$ be the polynomial of degree less than $d$ that is one modulo $p_j$ and zero modulo the other factors, as supplied by the Chinese remainder theorem. Let $R_j$ contain the coefficient rows of $z^k\psi_j(z)\bmod p_A(z)$ for $0\le k<\deg p_j$. \par
\emph{Groupwise estimates.} For the stable and unstable factors, every root has distance at least $\theta$ from the unit circle. Thus $|p_j(z)|\ge\theta^{\deg p_j}$ when $|z|=1$. Their degrees and coefficients are bounded in terms of $n,M$. Hence there is a constant \mbox{$\beta_0\ge1$}, depending only on $n,M,\theta$, such that \[ \max_{|z|=1}\frac{\|\phi_{p_j}(z)\|}{|p_j(z)|}\le\beta_0 \qquad(p_j=p_s\ \text{or}\ p_u). \] We claim that every row $v^*=v^*\psi_j(A)$ satisfies \begin{equation} \eps\|v\|\le\sqrt m\,\beta_0 \Bigl(\sum\nolimits_{k<\deg p_j}\|v^*A^kB\|^2\Bigr)^{1/2}. \label{eq:app-group} \end{equation} For a stable group, $v^*p_s(A)=0$, so the transfer identity \eqref{eq:app-transfer}, Parseval's inequality, and the record obstruction with $\tau=0$ give this estimate. Only the projected response needs to be stable. For a one-root group, the estimate follows from \eqref{eq:app-modal} and \mbox{$\beta_0\ge1$}. \par
For the unstable group, put $r_B=\eps/\sqrt m$. The error-free compatibility set contains $(A,B+\Delta B)$ for every real $\|\Delta B\|\le r_B$, because $\|\Delta B U\|\le\eps\sqrt T$. Let $(K,P)$ be its certificate and put $D=A+BK$. The argument at the end of this proof extends the same certificate to complex perturbations in this ball. Consequently, \begin{equation} \sup_{|z|=1}\|K(zI-D)^{-1}\|<1/r_B. \label{eq:app-actuator-gain} \end{equation} Indeed, if this bound fails at some $z$, choose a unit maximizing vector $x$, put $y=K(zI-D)^{-1}x$, and set $\Delta B=xy^*/\|y\|^2$. Then $\|\Delta B\|\le r_B$, while $D+\Delta BK$ has an eigenvalue at $z$, a contradiction. \par
For an unstable row, the resolvent identity gives \[ v^*=\frac1{2\pi i}\oint_{|z|=1} v^*(zI-A)^{-1}BK(zI-D)^{-1}\,dz. \] Here the projected row resolvent is continued after cancelling the other polynomial factors; any unit-circle poles of $A$ are removable in this row. The integrand equals \[ v^*\bigl((zI-D)^{-1}-(zI-A)^{-1}\bigr). \] Its first term integrates to $v^*$ because $D$ is Schur, and its second integrates to zero because the projected unstable resolvent is analytic inside the circle. Apply the transfer identity with $p_u$, the bound defining $\beta_0$, and \eqref{eq:app-actuator-gain} to the integral. The contour has length $2\pi$, so this proves \eqref{eq:app-group}. \par
\emph{Recombination.} Choose a real unit left singular vector $q$ of $\mathcal C_d$ for $\smin(\mathcal C_d)$, and let $H_q$ have rows $q^\top A^kB$, $0\le k<d$. For each group put \mbox{$q_j^*=q^\top\psi_j(A)$}. Polynomial congruences give $q^\top=\sum_jq_j^*$ and $[q_j^*A^kB]_{k<\deg p_j}=R_jH_q$. Hence \[ \begin{aligned} \eps\le\eps\sum_j\|q_j\| &\le\sqrt m\,\beta_0\sum_j\|R_jH_q\|_F\\ &\le\sqrt m\,\beta_0\,\smin(\mathcal C_d)\sum_j\|R_j\|. \end{aligned} \] The Chinese remainder coefficient systems are uniformly invertible. On the closure of each fixed-degree family, the roots remain bounded and different groups remain separated by $2\theta$. The inverse coefficient maps are therefore continuous on a compact set. Collisions within one group are allowed, and there are only finitely many degree partitions. This bounds the displayed coefficient by $\beta(n,M,b_0)$ without bounding any spectral projector in the original coordinates. Use $\smin(\mathcal C_d)\le\smin(\mathcal C)$, take the supremum over certified levels and causal experiments, and combine with the short-experiment lower bounds. \par
\emph{Complex perturbations.} It remains to justify the use of complex actuator perturbations. For the real $D,K,P$ and radius $r_B$ above, define, on real vectors, \[ \begin{split} f(x,w)&=(Dx+w)^\top P(Dx+w)-x^\top Px,\\ g(x,w)&=r_B^2\|Kx\|^2-\|w\|^2. \end{split} \] The condition $g\ge0$ means $w=EKx$ for some real $\|E\|\le r_B$. The real certificate and compactness on $\|x\|=1$ give $f\le-\alpha_*\|x\|^2$ on this constraint for some $\alpha_*>0$. If $K\ne0$, there is an $x$ with $Kx\ne0$, and $w=0$ makes $g>0$. The scalar S-lemma~\cite{polik2007} therefore gives $t\ge0$ such that $f+(\alpha_*/2)\|x\|^2+tg\le0$ for all real $x,w$. Its real symmetric matrix inequality also holds on complex vectors. Substituting $w=EKx$ then proves the required strict Hermitian inequality for every complex $\|E\|\le r_B$. If $K=0$, the conclusion is immediate. \qed\par
\subsection{Proof of Theorem~\ref{thm:duration}}\label{app:duration} \emph{Finite controllability.} The expansion \[ A_h^jB = (I+hG)^jB=\sum_{k=0}^j\binom jk h^kG^kB \] relates the blocks $B,A_hB,\ldots,A_h^{\nu-1}B$ to the full-row-rank matrix $[B\ GB\ \cdots\ G^{\nu-1}B]$ by a fixed upper-triangular matrix of binomial coefficients, with diagonal entries one, and the weights $1,h,\ldots,h^{\nu-1}$. Thus $\smin(\mathcal C)\ge\alpha_0h^{\nu-1}$. A real unit vector $q$ with $q^\top G^jB=0$ for $0\le j\le\nu-2$ has $q^\top A_h^jB=O(h^{\nu-1})$ for $j<n$, proving the upper estimate. \par
\emph{Causal obstruction.} Generate the true record with the disturbance of Section~\ref{sec:duration} and zero sensor error. Then $q^\top x_t=0$, and the same record is error-free for the unstabilizable alternative \eqref{eq:slow-alternative}. With $P_q=I-qq^\top$, the identities $(P_qG)^jB=G^jB$ for $j\le\nu-2$ give $q^\top G(P_qG)^jB=0$ for $j\le\nu-3$. More explicitly, the disturbed trajectory satisfies \[ x_t=\sum_{r=0}^{t-1}(I+hP_qG)^rB u_{t-1-r}. \] Expand its matrix powers and use $\sum_{r=j}^{t-1}\binom rj=\binom t{j+1}$. The vanished lower-order terms then give \[ \begin{split} \max_{t<T}\|w_t\| &\le\sqrt m\,\|B\|\sum_{\ell\ge\nu-1} \binom T\ell(h\|G\|)^\ell\\ &\le\frac{\sqrt m\,\|B\|}{(\nu-1)!} e^{hT\|G\|}(hT\|G\|)^{\nu-1}. \end{split} \] All columns lie in $\operatorname{span}\{q\}$, so $\|W\|\le\sqrt T\max_t\|w_t\|$. For $hT\le1$ this proves the claimed obstruction. For \mbox{$hT>1$}, the input-cancellation disturbance $w_t=-Bu_t$ keeps $x_t=0$, costs at most $\sqrt{mT}\|B\|$, and leaves $(I,0)$ compatible. This supplies the uniform constant ceiling in \eqref{eq:slow-causal}. \par
\emph{Attainment.} Let $N=\lceil1/h\rceil$ and $T=mN$, and hold each coordinate input for $N$ steps. Keep separate process and sensor error levels: $\delta_w=\eps_w\sqrt T$ and $\delta_v=\eps_v\sqrt T$. For the error-free record, put \mbox{$Z_h^0=[hX^0;U]$}. Euler convergence on the bounded rescaled interval and Riemann sums give \[ hZ_h^0(Z_h^0)^\top\longrightarrow\int_0^m[y(s);u(s)][y(s);u(s)]^\top\,ds\succ0, \] where $\dot y=Gy+Bu$, $y(0)=0$, and $u=e_j$ on the $j$th unit interval. If the quadratic integral vanishes at $(a,c)$, then $a^\top y+c^\top u=0$ almost everywhere. Analyticity makes this an identity on each interval. At the first interval's start \mbox{$c_1=0$}, and differentiation gives $a^\top G^kBe_1=0$ for every $k$. Its response is therefore annihilated on later intervals. Induction yields \mbox{$c_j=0$} and $a^\top G^kBe_j=0$ for all $j,k$, whence \mbox{$a=c=0$} by controllability. Thus $\sigma_{\min}(Z_h^0)\ge\alpha_1h^{-1/2}$. \par
Finite convolution gives \[ \|h(X-X^0)\|\le h\sum_{j<T}\|A_h^j\|\delta_w\le\beta_0\delta_w. \] Here shifts of the disturbance array have operator norm at most $\delta_w$, and $hT$ is bounded, so $\beta_0$ is independent of small $h$. Every compatible full underlying trajectory differs from the true one by at most $2\delta_v$. For sufficiently small $\eps_w+\eps_v$, therefore $\sigma_{\min}[hX';U]\ge\alpha_2h^{-1/2}$ uniformly over all compatible alternatives. With $G'=(A'-I)/h$ and $E_X=X'-X$ and $E_Y=Y'-Y$, subtraction gives \[ [G'-G\ \ B'-B][hX';U] =E_Y-E_X-hGE_X+W-W'. \] Consequently, \[ \|[G'-G\ \ B'-B]\| \le\frac{2\delta_w+(4+2h\|G\|)\delta_v}{\alpha_2h^{-1/2}} \le\beta_1(\eps_w+\eps_v). \] Choose a continuous-time stabilizing gain $K_{\mathrm{ct}}$ and Lyapunov matrix $P\succ0$ for $(G,B)$. If $\eps_w+\eps_v$ is sufficiently small, all $M'=G'+B'K_{\mathrm{ct}}$ obey one strict continuous-time Lyapunov inequality and remain bounded. Since \[ \begin{split} &(I+hM')^\top P(I+hM')-P\\ &\quad=h(M'^\top P+PM')+h^2M'^\top PM', \end{split} \] the gain $hK_{\mathrm{ct}}$ certifies every alternative for all sufficiently small $h$. Along $(\eps_w,\eps_v)=(\eps,\kappa\eps)$ this proves a constant lower tolerance. \par
This certificate can also be reconstructed from data and bounds. Put $h_N=1/N$, so $h/h_N=hN\to1$, and let $\widetilde Z_N=[h_N\widetilde X;U]$, $s=\sigma_{\min}(\widetilde Z_N)$, and $[\widehat G\ \widehat B]=(\widetilde Y-\widetilde X)\widetilde Z_N^\dagger$. The least-squares residual annihilates $\widetilde Z_N^\dagger$. The same subtraction argument therefore encloses every scaled compatible pair $((A'-I)/h_N,B')$ in the ball about $(\widehat G,\widehat B)$ of radius \[ \eta_h= \frac{\delta_w+(2+h_N\|\widehat G\|)\delta_v} {s-h_N\delta_v}. \] Its denominator is positive for small noise, $\eta_h=O(\eps_w+\eps_v)$, and its center is within $O(\eps_w+\eps_v+h)$ of $(G,B)$. Apply Lemma~\ref{lem:margin} to the algebraic coefficient-ball record $X_a=[I\ 0]$, $U_a=[0\ I]$, with successor matrix $Y_a=[I+h_N\widehat G\ \ h_N\widehat B]$ and radius $h_N\eta_h$. The synthetic set contains $(A',h_NB')$ for each physical compatible pair $(A',B')$. Thus its returned gain $K_a$ gives the physical gain $K=h_NK_a$, since $A'+(h_NB')K_a=A'+B'(h_NK_a)$. The preceding Lyapunov argument ensures a certificate uniformly in the stated regime. No true plant parameter is an input to this step. \par
The attaining sequence and the input-cancellation bound imply $\Gamma\asymp1$. Applying \eqref{eq:slow-causal} at the two fixed shortest durations and using the lower bounds of Theorem~\ref{thm:competitive}(i) gives $E_{\mathrm{exact}},E_{\mathrm{meas}}^\kappa\asymp h^{\nu-1}$. The duration conclusion and the stated stable-gap sharpness follow. \qed\par
\raggedbottom\section*{Acknowledgment} The authors acknowledge the use of GPT-6 Astra and Claude 5.5 during the preparation of this manuscript. Their use was integrated into the iterative drafting and revision process, making individual contributions difficult to isolate. All content was reviewed by the authors, who take full responsibility for the final manuscript. \par
 \par
\par
\par
\noindent\begin{minipage}{\columnwidth} \vspace{-2\baselineskip} \begin{IEEEbiography}[{\includegraphics[trim=14.4bp 0bp 14.4bp 0bp,clip,width=1in,height=1.25in]{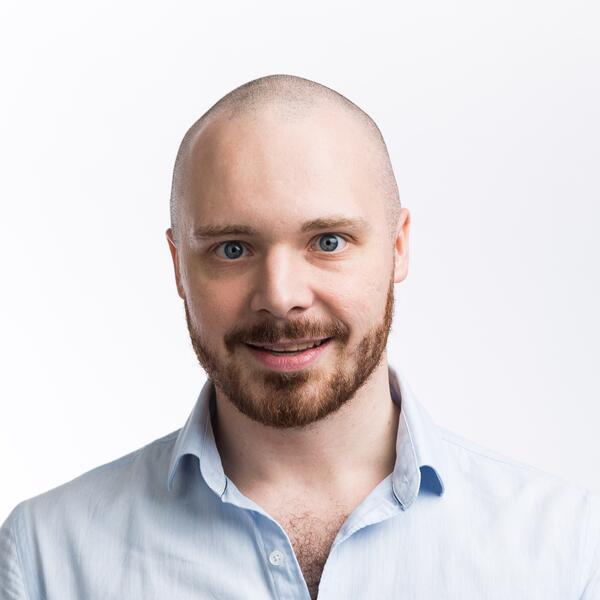}}]{Alexey Peregudin} received the B.Sc., M.Sc., and Ph.D. degrees from ITMO University, St. Petersburg, Russia, in 2015, 2017, and 2021, respectively. From 2017, he held academic and research positions at ITMO University, including Assistant Professor, Associate Professor, and Senior Researcher. Since 2023, he has been a University Teacher with the School of Electrical and Electronic Engineering, University of Sheffield, Sheffield, U.K. His interests include robust and optimal control, invariant-set methods, optimization, and data-driven control. \end{IEEEbiography} \par
\vspace{-2\baselineskip} \begin{IEEEbiography}[{\includegraphics[trim=0bp 212.5bp 0bp 0bp,clip,width=1in,height=1.25in]{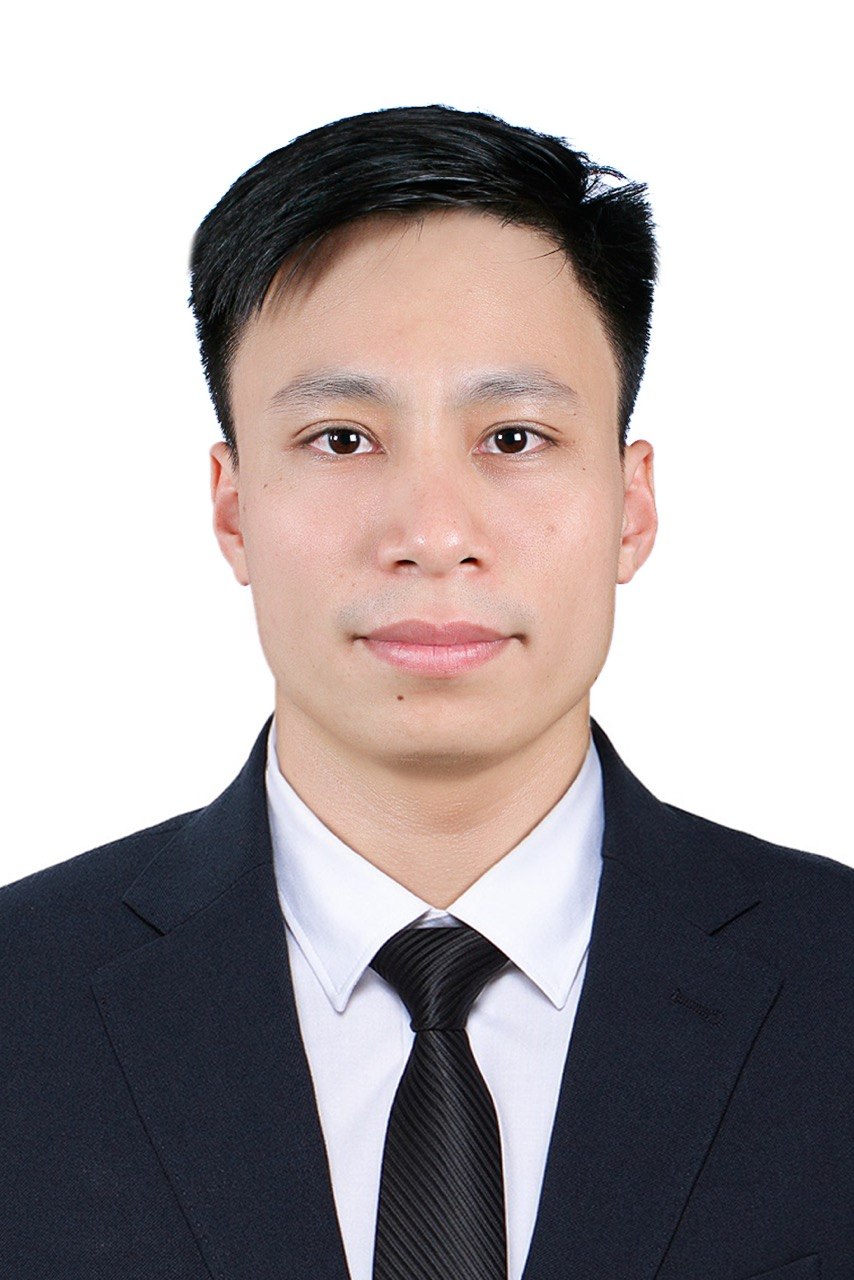}}]{Ngoc Tuan Dinh} received the B.Sc. and M.Sc. degrees in mechatronics and robotics from ITMO University, Saint Petersburg, Russia, in 2024 and 2026, respectively. Currently he is pursuing the Ph.D. degree at ITMO University. His research interests include data-driven control, robust and adaptive control, system identification, and analysis and control of uncertain dynamical systems. \end{IEEEbiography} \end{minipage} \end{document}